\documentclass[letterpaper,10 pt, conference]{ieeeconf}
\IEEEoverridecommandlockouts                               
\usepackage{fancyhdr}
\usepackage{hyperref}
\usepackage{xspace}
\usepackage[font={small}]{caption}
\usepackage{color}
\usepackage{setspace}
\usepackage{enumerate}
\usepackage{graphics}
\usepackage{cite}
\usepackage{dsfont}
\usepackage{latexsym}
\usepackage{amsmath}
\usepackage{amssymb}
\usepackage{amsfonts}
\usepackage{amsthm}
\usepackage{times}
\usepackage{sgame}
\usepackage{color}
\usepackage{verbatim}
\usepackage{xcolor}
\usepackage{pgfplots}
\usepackage{epsfig} 
\usepackage{subcaption}

\usepackage{hyperref}
\let\labelindent\relax
\usepackage{enumitem}
\newlist{inparaenum}{enumerate}{2}% allow two levels of nesting in an enumerate-like environment
\setlist[inparaenum]{nosep}% compact spacing for all nesting levels
\setlist[inparaenum,1]{label=\bfseries\arabic*.}% labels for top level
\setlist[inparaenum,2]{label=\arabic{inparaenumi}\emph{\alph*})}% labels for second level
\newcommand{\be}{\begin{equation}}
\newcommand{\ee}{\end{equation}}
\newcommand{\mbb}[1]{\mathbb{#1}}

\newcommand{\mcal}[1]{\mathcal{#1}}

\theoremstyle{plain}% default
\newtheorem{theorem}{Theorem}[section]

\newtheorem{lemma}{Lemma}[section]

\newtheorem{fact}{Fact}[section]

\newtheorem{assumption}{Assumption}
\newtheorem{example}{Example}
\newtheorem{definition}{Definition}

\def\bs{\boldsymbol}

\def\mce{\mcal{E}}
\def\mbeta{\beta}
\def\abeta{\beta_a}
\def\heal{\alpha}
\def\s{s}
\def\i{i}
\def\iIFE{\i_\text{IFE}}
\def\IFE{\bs{i}_\text{IFE}}
\def\ipeak{i_\text{pk}}

\normalsize\title{\LARGE \bf
	Analysis of Contagion Dynamics with Active Cyber Defenders
	\thanks{This work was supported in part by %CSI Cybersecurity Initiative 
Colorado State Bill 18-086, and NSF Grants \#2122631 and \#2115134. The authors are with the Department of Computer Science at University of Colorado, Colorado Springs. Contact: \texttt{\{kpaarpor,sxu\}@uccs.edu}}}

\author{
	Keith Paarporn, Shouhuai Xu
}

\begin{document}

\maketitle

\begin{abstract}

In this paper, we analyze the infection spreading dynamics of malware in a population of cyber nodes (i.e., computers or devices). Unlike most prior studies where nodes are reactive to infections, in our setting some nodes are \emph{active defenders} meaning that they are able to clean up malware infections of their neighboring nodes, much like how spreading malware exploits the network connectivity properties in order to propagate. We formulate these dynamics as an Active Susceptible-Infected-Susceptible (A-SIS) compartmental model of contagion. We completely characterize the system's asymptotic behavior by establishing conditions for the global asymptotic stability of the infection-free equilibrium and for an endemic equilibrium state. We show that the presence of active defenders counter-acts infectious spreading, effectively increasing the epidemic threshold on parameters for which an endemic state prevails. Leveraging this characterization, we investigate a general class of problems for finding optimal investments in active cyber defense capabilities given limited resources. We show that this class of problems has unique solutions under mild assumptions. We then analyze an  Active Susceptible-Infected-Recovered (A-SIR) compartmental model, where the peak infection level of any trajectory is explicitly derived.
    
    % The presence of active defenders balances the existing asymmetries that favor cyber attackers.
    
    % The spread of malware through a cyber network draws parallels with and is often described by the very models that characterize epidemic dynamics of infectious diseases. In these dynamics, the cyber malware (the disease) is able to exploit the connectivity of the network in order to propagate. 

\end{abstract}

%%%%%%%%%%%%%%%%%%%%%%%%%%%%%%%%%%%%%%%%%%%%%%%%%%%%%%%
\section{Introduction}
% !TEX root = main.tex

% Information and cyber networks are essential to the operation of modern engineering systems. The interconnected nature of these systems enable data to be transmitted between numerous network nodes. However, these interconnections also introduce vulnerabilities to cyber attacks. The spread of malware and computer viruses are able to exploit the connectivity properties of the network in order to propagate. These threats are still a major cause for concern, despite breakthroughs in certain areas of cybersecurity such as cryptographic algorithms, intrusion detection, firewalls, and anti-malware. 
The spread of computer malware and viruses remains a major cause for concern, despite the tremendous amount of effort by academia, industry, and government.
This is true despite the substantial progresses in %%breakthroughs in 
certain areas of cybersecurity such as cryptography, intrusion detection, firewalls, and anti-malware tools. 
These traditional cyber defense approaches are {\em preventive} and {\em reactive} in nature because they strive to prevent attacks from succeeding and react to recognized attacks \cite{xu2014adaptive,zheng2017preventive}. However, it is known that cyber attacks cannot be be completely prevented, for reasons that include undecidability \cite{adleman1990abstract} and human factors \cite{montanez2020human}. Moreover, reactive defenses are limited because there may be substantial delays before attacks are detected and cleaned up.

% for reasons that include undecidability (e.g., there does not exist any universal solution to determining whether a piece of software is malicious or not \cite{Adleman-Crypto-88}) and human factors (e.g., humans are inherently vulnerable to cyber social engineering attacks \cite{XuFrontierInPsychology2020}).

The limitation of traditional defenses is characterized by an asymmetry that benefits attackers. Namely, the effect of attacks is amplified by the network's connectivity (malware spreading), but the effect of preventive and reactive defenses are not \cite{xu2014adaptive,zheng2017preventive,xu2015stochastic,lu2013optimizing,zheng2015active}. This asymmetry has led to an emerging class of countermeasures called {\em active cyber defenses} \cite{xu2015stochastic,lu2013optimizing,zheng2015active,theron2019autonomous,theron2020reference}, which leverage the same interconnections exploited by attacks to actively identify and clean up compromised nodes. This is achieved by endowing uncompromised nodes the ability to ``hunt" compromised nodes to clean up their infection status (or ``remotely delivering cures'').

% letting the {\em secure} nodes in a network {\em actively} hunt attacks (i.e., {\em compromised} nodes in the network) and {\em automatically clean up the compromises}, which neutralizes the asymmetry that benefits the attacker as the attacks automatically spread over the network.  
% From the practical point of view, full-fledged active cyber defenses have yet to be realized because there are a set of open problems ranging from the legal aspect to the technological realization aspect \cite{XuAICA2022}. Nevertheless, the potential of active cyber defenses are well recognized and it's certainly important to characterize what they can accomplish once they are fully realized.  

%much like cyber attackers. These effectively eliminate the aforementioned asymmetry between attackers and defenders. The study of active cyber defenses in the context of control and dynamical systems has received attention in recent years. 

This paper investigates the impact that active defenders have on the spread of malware. The spreading dynamics of malware draws parallels to the spread of an infectious disease in a human population, and as such, basic models in epidemiology are often utilized to study cybersecurity dynamics \cite{ganesh2005effect,van2008virus,zheng2017preventive,han2021preventive,varma2022non,mai2022optimal}. The dynamics of these models often follow an epidemic threshold, such that parameter instances that lie below the threshold characterize an infection-free equilibrium, and instances above the threshold exhibits an endemic equilibrium, where a constant fraction is infected over long periods of time. In this paper, we formulate cybersecurity dynamics as the A-SIS compartmental model, where nodes are either {\em susceptible} to, or {\em infected} by, the malware, and some
%a fixed fraction of the 
nodes are {\em active} defenders.

There have been studies on characterizing the effectiveness of active defenders, mainly from a holistic (i.e., network-oriented) perspective \cite{xu2015stochastic,lu2013optimizing,zheng2015active}. A main focus has been on studying their advantage over preventive and reactive cyber defenses \cite{xu2015stochastic}, and their potential side-effects \cite{zheng2015active}. These characterization works study mean-field approximations and provide conditions for the existence or absence of equilibria. A prescription study has derived optimal control strategies for active cyber defenses in a population contagion model \cite{lu2013optimizing}, where it is assumed that {\em every} node in the network has active defense capabilities and no nodes are equipped with reactive defenses. This motivates us to consider more realistic scenarios, where nodes are equipped with reactive defense capabilities and some nodes can be equipped with active defense capabilities. From a practical standpoint in terms of materializing the potential of active defenses in the real world, ongoing work explores competent architectures for implementing active defense approaches \cite{theron2020reference}, and systematizing the challenges that must be tackled before its full potential can be realized \cite{XuAICA2022}.

In relation to the literature on epidemics, our work is similar to the dynamics of competitive bi-virus models \cite{liu2019analysis,santos2015bi,doshi2022bi}, though the mathematical equations differ in two respects. First, there are only two compartments (susceptible and infected) in our model of active defenses, whereas bi-virus models must account for three compartments. Second, active cyber defenses have inherent asymmetries, since it is possible that only a fraction of the nodes may implement them. In bi-virus models, any node that is infected with a given virus type may spread it.

\vspace{1mm}

\noindent\textbf{Contributions:} Our study abstracts away the underlying complex network structures originally considered in \cite{xu2015stochastic,zheng2015active}. Our A-SIS (Active SIS) model is based on a well-mixed  population of nodes, where each node has the same rate of interaction with others. Under these assumptions, our paper provides a full characterization of the dynamical properties of the A-SIS contagion model. We precisely characterize its epidemic threshold, which increases in the fraction and effectiveness of active defenders (Theorem \ref{thm:IFE_GAS}). We note that this is in contrast to many other studies of epidemic models with reactive population behaviors, i.e. non-pharmaceutical interventions such as social distancing, where the epidemic threshold remains unchanged from classic non-behavioral models \cite{paarporn2017networked,eksin2019systematic,reluga2010game}. Additionally, we fully characterize the global stability properties of the equilibrium states of the system by identifying suitable Lyapunov functions (Section \ref{sec:stability}). Specifically, these equilibria are the infection-free equilibrium (IFE) and the endemic equilibrium states, where we provide precise infection levels for the latter. Based on these characterizations, we then consider how a designer should optimize system security by investing monetary assets in increasing the fraction of active defense nodes and their effectiveness (Section \ref{sec:investment}). We show that there is a unique optimal investment profile, given concave return functions.

%%%%%%%%%%%%%%%%%%%%%%%%%%%%%%%%%%%%%%%%%%%%%%%%%%%%%%%
\section{The A-SIS contagion model}

Our model of active cyber defense is built upon the classic compartmental SIS model, which we first review below. 

\subsection{The SIS Epidemic Model}

A malware spreads through a population of nodes, where each node is either susceptible or infected with the malware.. We denote $\s(t) \in [0,1]$ as the fraction of nodes in the network that are susceptible at time $t$, and $\i(t) \in [0,1]$ as the fraction that are infected. The malware can only be transmitted from an infected node to a susceptible node upon contact. The per-contact infection rate is $\mbeta > 0$. Infected nodes are able to independently recover at the rate $\heal > 0$ by using reactive defenses, e.g. by using recovery software, intrusion-detection system, or anti-malware tool. The states $\s$ and $\i$ evolve according to the following dynamical system:

\begin{equation}
    \begin{aligned}
        \frac{d\s}{dt} &= -\mbeta \s\i + \heal \i \\
        \frac{d\i}{dt} &= \mbeta \s\i - \heal \i 
    \end{aligned}
\end{equation}
Under these dynamics, the mass of the population is invariant, and we must have $s(t) + i(t) = 1$ for all times $t$. Therefore, the SIS dynamics may be reduced to a single state variable,
\begin{equation}\label{eq:SIS_dynamics}
    \frac{d\i}{dt} = \mbeta \i(1-\i) - \heal \i 
\end{equation}
with initial condition $i(0) \in [0,1]$. The \emph{infection-free equilibrium} (IFE) $\iIFE = 0$ is an equilibrium of the system. The \emph{endemic} equilibrium $i^* = 1 - \frac{\heal}{\mbeta} \in (0,1]$ is an equilibrium of the system if and only if $\frac{\mbeta}{\heal} > 1$. The solution to this system is well-known. Its asymptotic properties are summarized.
\begin{theorem}[\cite{mei2017dynamics}]\label{thm:SIS}
    Consider the SIS model \eqref{eq:SIS_dynamics}.
    \begin{itemize}
        \item If $\frac{\mbeta}{\heal} \leq 1$, then $\i(t)$ converges to $\iIFE$ for any initial condition $i(0) \in [0,1]$.
        \item If $\frac{\mbeta}{\heal} > 1$, then $\i(t)$ converges to the \emph{endemic state} $i^* = 1 - \heal/\mbeta$ for any initial condition $i(0) \in (0,1]$. 
    \end{itemize}
\end{theorem}
From the above result, the dynamics of the SIS model follow a threshold on the parameter $\frac{\beta}{\alpha}$, which is often referred to as the basic reproduction number.
% The number $R_0 = \frac{\mbeta}{\heal} > 0$ is known as the basic reproduction number.

%%%%%%%%%%%%%%%%%%%%%%%%%%%%%%%%%%%%%%%%%%%%%%%%%%%%%%%
\subsection{A-SIS: Active cyber defense dynamics}

We now formulate our dynamical model of active cyber defense. Like in the SIS model, every node comes equipped with recovery software (and thus recover from infection at rate $\heal$). Now, suppose a fixed fraction $x_a \in [0,1]$ of the nodes are \emph{active defenders}. In addition to reactive defenses, they are able to employ active defenses against the spreading malware. Upon contact with any infected node, an active susceptible node is able to clean up the malware at the infected node with rate $\abeta>0$. We assume that an active defender can only clean up other infected nodes when it is itself not infected, which is natural. Thus, a susceptible state is required to implement active defenses. The remaining fraction $1-x_a$ of the nodes are not active defenders. We refer to these nodes as \emph{reactive} nodes.

\begin{figure}
    \centering
    \includegraphics[scale=0.3]{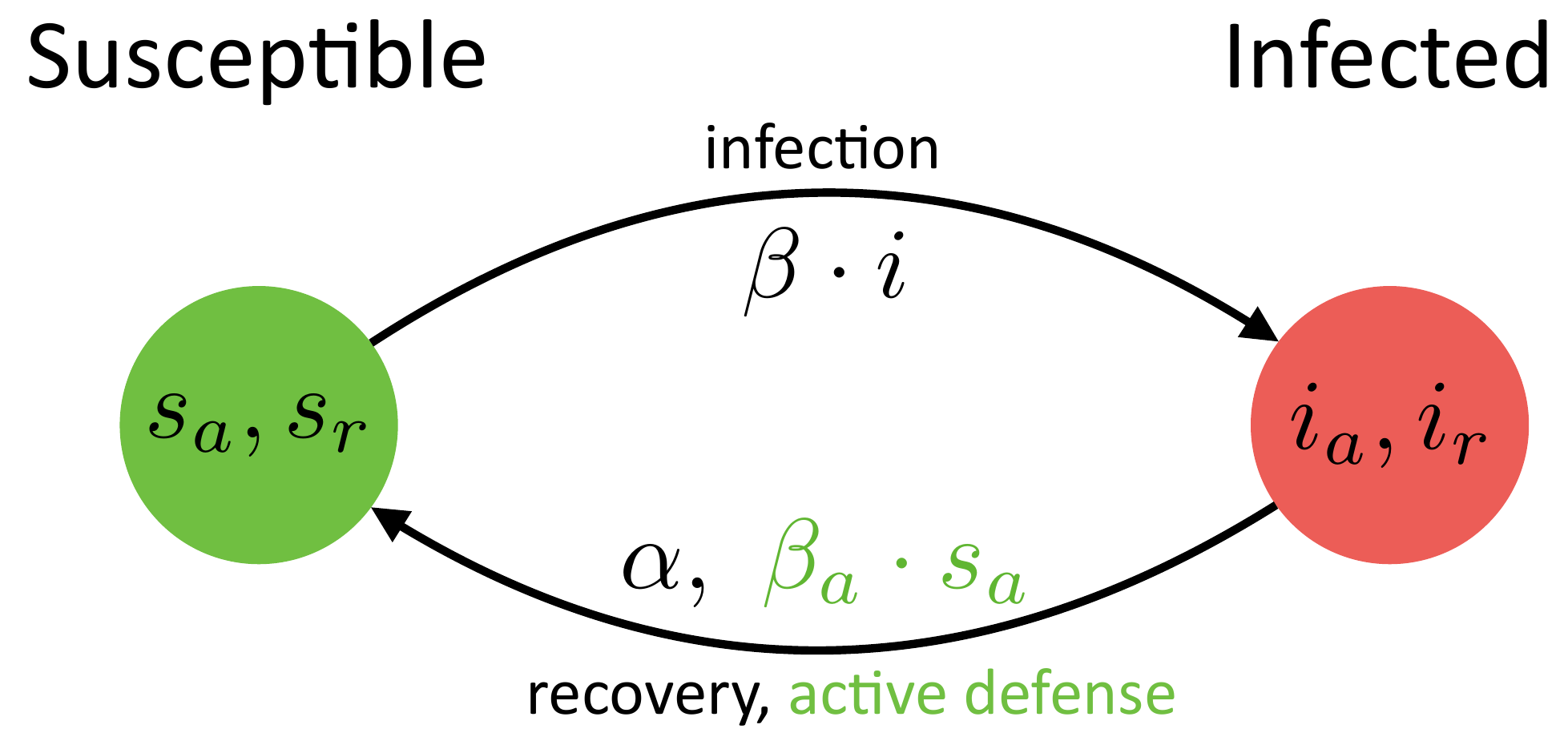}
    \caption{State transition diagram under active cyber defense dynamics (A-SIS). Each node is one of two types that describe its equipped defense technology. Active nodes (subscript $a$) have both active and reactive defenses, whereas reactive nodes (subscript $r$) have only recovery capability. An infected node can transition to susceptible either through traditional recovery, or by interactions with active defenders that are susceptible.}
    \label{fig:main_figure}
\end{figure}

Let $\s_a(t)$, $\s_r(t)$, $\i_a(t)$, and $\i_r(t)$ be the fraction of susceptible active, susceptible reactive, infected active, and infected reactive nodes, respectively. We denote $\i(t) \triangleq \i_a(t) + \i_r(t)$ as the total infected fraction at time $t$. The dynamics are given by
\begin{equation}\label{eq:SIS_active_full_dynamics}
    \begin{aligned}
        \frac{d\s_a}{dt} &= -\mbeta \s_a \i + \abeta \s_a \i_a + \heal \i_a \\
        \frac{d\i_a}{dt} &= \mbeta \s_a \i - \abeta \s_a \i_a - \heal \i_a \\
        \frac{d\s_r}{dt} &= -\mbeta \s_r \i + \abeta \s_a \i_r + \heal \i_r \\
        \frac{d\i_r}{dt} &= \mbeta \s_r \i - \abeta \s_a \i_r - \heal \i_r \\
    \end{aligned}
\end{equation}
A state transition diagram is shown in Figure \ref{fig:main_figure}. From the above equations, $\frac{d\s_a}{dt} + \frac{d\i_a}{dt} = 0$ and $\frac{d\s_r}{dt} + \frac{d\i_r}{dt} = 0$, and therefore we have $\s_a(t) + \i_a(t) = x_a$ and $\s_r(t) + \i_r(t) = 1-x_a$ for any time $t\geq 0$. The dynamics in \eqref{eq:SIS_active_full_dynamics} is thus a planar system with state $\bs{i} \triangleq (\i_a,\i_r) \in [0,x_a]\times[0,1-x_a]$ governed by the dynamics
\begin{equation}\label{eq:SIS_active_dynamics}
    \begin{aligned}
        \frac{d\i_a}{dt} &= F_a(\bs{i}) \triangleq \mbeta (x_a - \i_a) \i - \abeta (x_a-\i_a) \i_a - \heal \i_a \\
        \frac{d\i_r}{dt} &= F_r(\bs{i}) \triangleq \mbeta (1-x_a-\i_r) \i - \abeta (x_a-\i_a) \i_r - \heal \i_r \\
    \end{aligned}\tag{A-SIS}
\end{equation}
We will denote the state space as $\Gamma \triangleq [0,x_a]\times[0,1-x_a]$. The initial condition is specified by $\bs{i}_0 = (\i_a(0),\i_r(0)) = (\i_{a0}, \i_{r0}) \in \Gamma$. We seek to characterize the equilibria of system \eqref{eq:SIS_active_dynamics} and their stability properties. We immediately see that the infection-free equilibrium $\bs{i}_\text{IFE} \triangleq (0,0)$ is an equilibrium point of system \eqref{eq:SIS_active_dynamics}. We say an equilibrium $\bs{i}^*$ is \emph{interior} if $\bs{i}^* \in (0,x_a)\times(0,1-x_a)$. We will study the stability properties of A-SIS in the next section.

%%%%%%%%%%%%%%%%%%%%%%%%%%%%%%%%%%%%%%%%%%%%%%%%%%%%%%%
\section{Stability analysis of A-SIS dynamics}\label{sec:stability}

We consider the following stability notions. Let $\mce = \{\bs{i} \in \Gamma : F(\bs{i}) = (0,0)\}$ be the set of equilibrium points.
\begin{definition}
    An equilibrium point $\bs{i}^* \in \Gamma$ is \emph{globally asymptotically stable (GAS) with respect to} $\Gamma$ if for all $\bs{i}_0 \in \Gamma\backslash\mce$, $\bs{i}(t)$ converges to $\bs{i}^*$.
\end{definition}
Since we are considering the particular system \eqref{eq:SIS_active_dynamics}, we will simply say an equilibrium point is globally asymptotically stable (GAS). The stability properties of $\IFE$ is summarized in our main result below.
\begin{theorem}\label{thm:IFE_GAS}
    Consider system \eqref{eq:SIS_active_dynamics}. 
    \begin{enumerate}
        \item The equilibrium $\IFE$ is globally asymptotically stable  if and only if $\frac{\mbeta}{\heal} \leq 1 + \frac{\abeta x_a}{\heal}$.
        \item When $\frac{\mbeta}{\heal} > 1 + \frac{\abeta x_a}{\heal}$, there exists a unique interior equilibrium $\bs{i}^* = (x_a\frac{\lambda_+}{\lambda_+ + \heal}, (1-x_a)\frac{\lambda_+}{\lambda_+ + \heal})$ that is globally asymptotically stable, where
        \begin{equation}
            \lambda_+ \triangleq \mbeta - \abeta x_a - \heal.
        \end{equation}
        Here, $\bs{i}^*$ is referred to as the endemic equilibrium.
    \end{enumerate} 
\end{theorem}
Note that Theorem \ref{thm:IFE_GAS} degenerates to Theorem \ref{thm:SIS} by setting $x_a=0$, meaning that all nodes only have reactive defenses. The presence of active defenders increases the threshold for which the IFE is GAS in comparison to the condition in Theorem \ref{thm:SIS}.

\subsection{Global asymptotic stability of the IFE}

Before establishing the GES result of $\IFE$, we first establish GAS of $\IFE$ in this subsection. The Jacobian is
\begin{equation}
    J(\bs{i}) =
    \begin{bmatrix}
        \frac{\partial F_a}{\partial\i_a} & \frac{\partial F_a}{\partial\i_r} \\
        \frac{\partial F_r}{\partial\i_a} & \frac{\partial F_r}{\partial\i_r}
    \end{bmatrix}
\end{equation}
where
\begin{equation}
    \begin{aligned}
        \frac{\partial F_a}{\partial\i_a} &= (\mbeta-\abeta)(x_a-2\i_a)-\mbeta \i_r - \heal \\
        \frac{\partial F_a}{\partial\i_r} &= \mbeta(x_a-\i_a) \\
        \frac{\partial F_r}{\partial\i_a} &= \abeta \i_r + \mbeta(1-x_a-\i_r) \\
        \frac{\partial F_r}{\partial\i_r} &= \mbeta(1-x_a-2\i_r-\i_a)-\abeta(x_a-\i_a)-\heal
    \end{aligned}
\end{equation}
Evaluated at $\IFE$, we have
\begin{equation}
    J(\IFE) =
    \begin{bmatrix}
        (\mbeta-\abeta)x_a - \heal & \mbeta x_a \\
        \mbeta(1-x_a) & \mbeta(1-x_a)-\abeta x_a - \heal
    \end{bmatrix}
\end{equation}
The characteristic equation is
\begin{equation}
    \begin{aligned}
        &(\lambda - ((\mbeta-\abeta)x_a - \heal))(\lambda - (\mbeta(1-x_a)-\abeta x_a - \heal)) \\
        &-\mbeta^2 x_a(1-x_a) = 0
    \end{aligned}
\end{equation}
and the eigenvalues are given by
\begin{equation}
    \lambda_+ \triangleq \mbeta - \abeta x_a - \heal, \quad \lambda_- \triangleq -(\abeta x_a + \heal).
\end{equation}
Since $\lambda_- < 0$, it is required that $\lambda_+ < 0$ for the IFE to be locally stable. This is the case under the condition
\begin{equation}
    x_a > \frac{\mbeta - \heal}{\abeta}.
\end{equation}
If $x_a < \frac{\mbeta - \heal}{\abeta}$, then the IFE is unstable.

Now, let us consider the two nullclines of system \eqref{eq:SIS_active_dynamics},
\begin{equation}
    \begin{aligned}
        \mcal{I}_a &\triangleq \left\{ \bs{i} \in \Gamma : F_a(\bs{i}) = 0 \right\} \\
        &= \left\{ \bs{i} \in \Gamma : \left(\frac{\heal}{\mbeta(x_a-\i_a)} - \left(1 - \frac{\abeta}{\mbeta} \right) \right) \i_a = \i_r \right\} \\
        \mcal{I}_r &\triangleq \left\{ \bs{i} \in \Gamma : F_r(\bs{i}) = 0 \right\} \\
        &= \left\{ \bs{i} \in \Gamma : \i_a = \left( \frac{\heal - \mbeta(1-x_a-\i_r) + \abeta x_a}{\mbeta(1-x_a-\i_r)+\abeta \i_r} \right) \i_r \right\}. \\
    \end{aligned}
\end{equation}
The intersection of the nullclines yields the set of equilibrium points. Observe that $\bs{i}_{\text{IFE}}$ is always an equilibrium. We may define functions
\begin{equation}
    \begin{aligned}
        I_a(\i_a) &\triangleq \left(\frac{\heal}{\mbeta(x_a-\i_a)} - \left(1 - \frac{\abeta}{\mbeta} \right) \right) \i_a \\
        I_r(\i_r) &\triangleq \left( \frac{\heal - \mbeta(1-x_a-\i_r) + \abeta x_a}{\mbeta(1-x_a-\i_r)+\abeta \i_r} \right) \i_r
    \end{aligned}
\end{equation}
whose graphs $(i_a,I_a(i_a))$ and $(I_r(i_r),i_r)$ are the $a$- and $r$-nullclines, respectively. The following convexity properties hold.

\begin{lemma}\label{lem:isocline_convexity}
Function $I_a(i_a)$ for $i_a \in [0,x_a)$ is convex and strictly increasing in $i_a$. Function $I_r(i_r)$ for $i_r \in [0,1-x_a)$ is convex in $i_r$. It is strictly increasing on $i_r \in [0,1-x_a)$ if $x_a \geq \frac{\mbeta-\heal}{\mbeta+\abeta}$, and on $i_r \in [\frac{\mbeta(1-x_a)-\abeta x_a-\heal}{\mbeta},1-x_a)$ if $x_a < \frac{\mbeta-\heal}{\mbeta+\abeta}$.
\end{lemma}

Note that $I_a(i_a)$ ($i_a \in [0,x_a]$) and $I_r(i_r)$ ($i_r \in [0,1-x_a]$) are functions defined on different domains. It will be more convenient to have them on the same domain, $i_a \in [0,x_a]$. Such a representation for the $r$-isocline is given below.

\begin{lemma}\label{lem:ia_representation}
    The $r$-nullcline can explicitly be represented as a function of $i_a\in[0,x_a]$ with
    \begin{equation}\label{eq:ia_representation}
        \begin{aligned}
            &\hat{I}_r(i_a) 
            \triangleq\\
            &\frac{1}{2}\left[-\left(\frac{d + i_a(\mbeta-\abeta)}{\mbeta}\right) + \right. \\
            &\quad\quad + \left. \sqrt{\left(\frac{d + i_a(\mbeta-\abeta)}{\mbeta} \right)^2 + 4i_a(1-x_a)} \right]
        \end{aligned}
    \end{equation}
    where $d \triangleq \heal + \abeta x_a - \mbeta(1-x_a)$.
\end{lemma}
\begin{proof}
    To obtain this representation, we must solve the equation $I_r(i_r) = i_a$ for $i_r$. After rearranging terms, it becomes a quadratic equation in $i_r$. Applying the quadratic formula, we obtain \eqref{eq:ia_representation}.
\end{proof}
The representation $\hat{I}_r(i_a)$ is the inverse function of $I_r(i_r)$. Because $I_r(i_r)$ is convex and strictly increasing with respect to $i_r$, it follows that $\hat{I}_r(i_a)$ is concave and strictly increasing in $i_a \in [0,x_a]$. These properties allow us to establish conditions for the existence of a unique interior equilibrium.

\begin{lemma}\label{lem:no_intersection}
    If $\lambda_+ \leq 0$, then $\bs{i}_{\text{IFE}}$ is the only isolated equilibrium in $\Gamma$. A unique interior fixed point exists if and only if $\lambda_+ >0$. 
\end{lemma}
\begin{proof}
    We will consider two cases.
    
    When $\lambda_+ \leq 0$, we have that $I_a(0) = \hat{I}_r(0) = 0$, and can directly calculate
    \begin{equation}
        \frac{\partial I_a}{\partial i_a} = \left( \frac{\heal}{\mbeta(x_a-i_a)} - \left(1 - \frac{\abeta}{\mbeta} \right) \right) + i_a\frac{\heal}{\mbeta(x_a-i_a)^2}
    \end{equation}
    to obtain
    \begin{equation}
        \frac{\partial I_a}{\partial i_a}(0) = \frac{\heal}{\mbeta x_a} - \left(1 - \frac{\abeta}{\mbeta} \right) \leq \frac{1-x_a}{x_a}.
    \end{equation}
    Similarly, we calculate 
    \begin{equation}
        \begin{aligned}
        &\frac{\partial \hat{I}_r}{\partial i_a} = \\
        &\frac{1}{2}\left[\frac{\abeta}{\mbeta}-1 + \frac{2(\mbeta-\abeta)(d+i_a(\mbeta-\abeta))/\mbeta^2 + 4(1-x_a)}{2\sqrt{(d+i_a(\mbeta-\abeta))^2/\mbeta^2 + 4i_a(1-x_a) }} \right]
        \end{aligned}
    \end{equation}
    and consequently obtain
    \begin{equation}
        \frac{\partial \hat{I}_r}{\partial i_a}(0) = \frac{\mbeta(1-x_a)}{\heal+\abeta x_a-\mbeta(1-x_a)} \geq \frac{1-x_a}{x_a}. 
    \end{equation}
    Due to the convexity of $I_a$ and concavity of $\hat{I}_r$, the only point at which the nullclines can intersect is the IFE.

    When $\lambda_+ > 0$, the system possesses a unique interior equilibrium, which we can solve for algebraically. We seek the non-zero solution of the system of equations $[F_a(i_a,i_r), F_r(i_a,i_r)]^\top = [0, 0]^\top$. We obtain the unique solution
    \begin{equation}
        i_a^* = x_a\cdot\frac{\lambda_+}{\lambda_+ + \heal}, \quad i_r^* = (1-x_a)\cdot\frac{\lambda_+}{\lambda_+ + \heal}.
    \end{equation}
\end{proof}

\begin{figure}[t]
    \centering
    \includegraphics[scale=0.25]{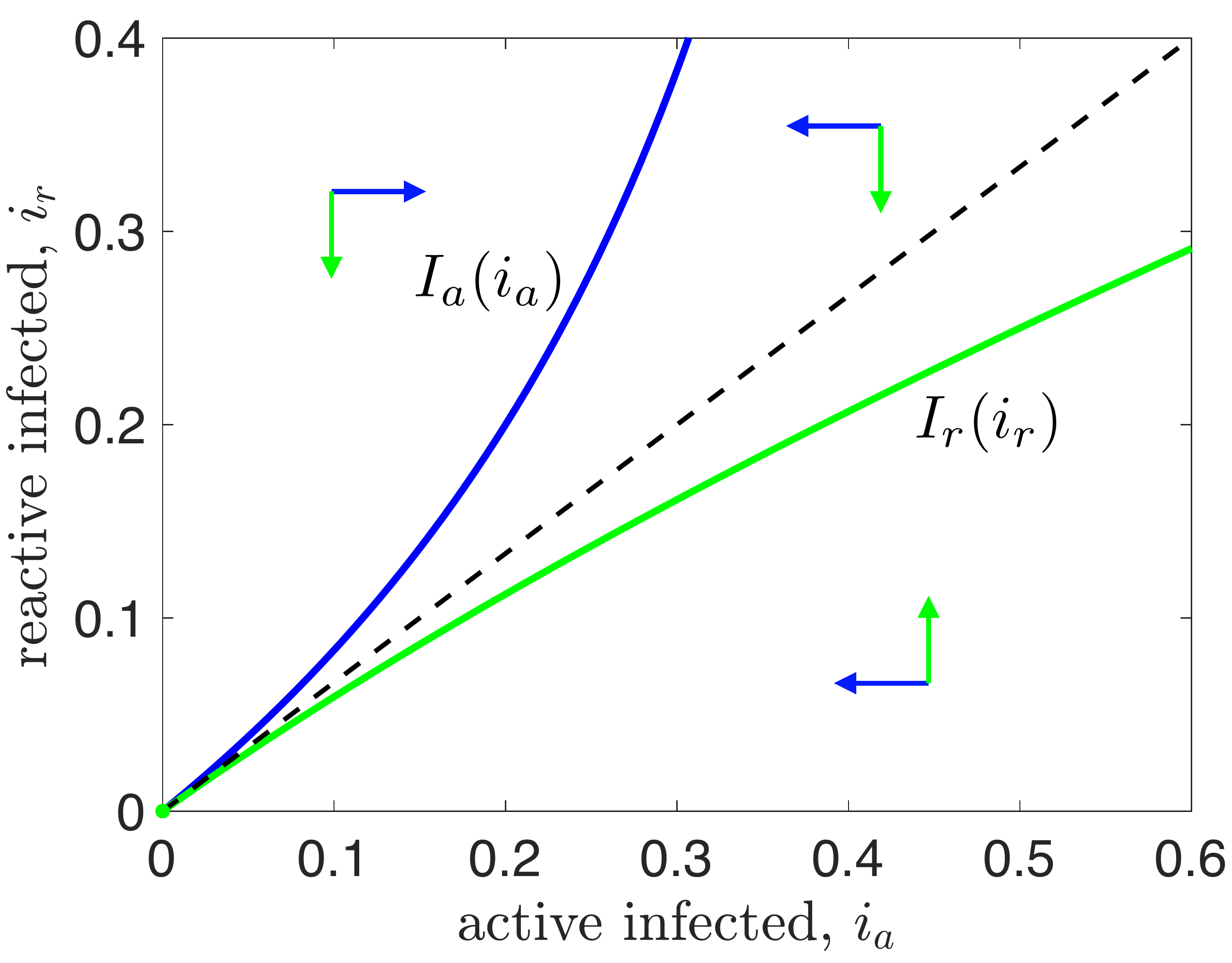}
    \caption{Nullclines of system \eqref{eq:SIS_active_dynamics}. The arrows depict the directions of the flow of system \eqref{eq:SIS_active_dynamics}. The dashed line is $i_r = \frac{1-x_a}{x_a}i_a$. Here, we have set $\mbeta = 0.3$, $\abeta = 0.35$, $\heal = 0.1$, and $x_a = 0.6$.}
    \label{fig:nullclines}
\end{figure}

Now, let us consider the candidate Lyapunov function $V:\Gamma \rightarrow \mathbb{R}_{\geq 0}$ defined by
\begin{equation}\label{eq:lyapunov_IFE}
    V(\bs{i}) \triangleq \max\left\{\frac{1-x_a}{x_a}i_a, i_r \right\}.
\end{equation}
It holds that $V(\bs{i}) \geq 0$ for all $\bs{i}$, with equality if and only if $\bs{i} = \IFE$. Observe that when $x_a > \frac{\mbeta-\heal}{\abeta}$ ($\lambda_+ < 0$), the line $\frac{1-x_a}{x_a}i_a$ is a lower bound for $I_a(i_a)$ and an upper bound for $\hat{I}_r(i_a)$ (see proof of Lemma \ref{lem:no_intersection}). Therefore, when $\lambda_+ < 0$, we have
\begin{equation}
    \frac{dV}{dt} = \mathds{1}_{\left\{i_r > \frac{1-x_a}{x_a}i_a\right\}} F_r(\bs{i}) + \mathds{1}_{\left\{i_r \leq \frac{1-x_a}{x_a}i_a\right\}} F_a(\bs{i}) \leq 0
\end{equation}
with equality if and only if $\bs{i} = \IFE$. Hence, $V(\bs{i})$ serves as a Lyapunov function that establishes the global asymptotic stability (w.r.t. $\Gamma$) of $\IFE$. This concludes the proof of Theorem \ref{thm:IFE_GAS} part 1.

A few remarks regarding the choice of $V$: the function $V(\bs{i})$ is of the form of a state max-separable Lyapunov function \cite{coogan2019contractive}. Therein, sufficient conditions on the Jacobian matrix on a forward-invariant convex set (here, $\Gamma$) are specified for which $V$ serves as a Lyapunov function of the equilibrium point. Specifically,
\begin{fact}[Corollary 5 from \cite{coogan2019contractive}]\label{fact:max_separable}
    Suppose there exists a vector $w > 0$ such that $J(\bs{i})w \leq 0$ for all $\bs{i} \in \Gamma$, and $J(\IFE)w < 0$. Then
    \begin{equation}
        \max\left\{\frac{1}{w_1}|i_a|, \frac{1}{w_2}|i_r|\right\}
    \end{equation}
    serves as a Lyapunov function and $\IFE$ is GAS w.r.t $\Gamma$.
\end{fact}
The choice of max-separable Lyapunov function $V$ in \eqref{eq:lyapunov_IFE} is induced from vector $w = [\frac{x_a}{1-x_a},1]^\top$. However, the sufficient conditions of Fact \ref{fact:max_separable} are not necessarily met. In particular, the second entry of $J((x_a,0))w$ is given by $\mbeta - \mbeta x_a - \heal$, which can be positive even when $\lambda_+ < 0$.

\subsection{Global stability of the endemic equilibrium}

The system possesses a unique interior equilibrium when $x_a < \frac{\mbeta-\heal}{\abeta}$ ($\lambda_+ > 0$) (Lemma \ref{lem:no_intersection}), which is given by
\begin{equation}
    i_a^* = x_a\cdot\frac{\lambda_+}{\lambda_+ + \heal}, \quad i_r^* = (1-x_a)\cdot\frac{\lambda_+}{\lambda_+ + \heal}.
\end{equation}

% % % Calculations
% At this point, it is also the case that $\frac{d}{dt}(i_a + i_r) = \frac{di}{dt} = 0$. We obtain
% \begin{equation}
%     \begin{aligned}
%         &\mbeta s i - \abeta s_a i - \heal i = 0 \\
%         &\Rightarrow \mbeta s - \abeta s_a - \heal = 0
%     \end{aligned}
% \end{equation}
% where we divided by $i > 0$ in the second line above. Using the substitutions $s = s_a + s_r$, $s_a = x_a - i_r$, and $s_r = 1-x_a-i_r$, we obtain the relation
% \begin{equation}\label{eq:ia_ir_relation}
%     i_a = \frac{1}{\abeta-\mbeta}(\mbeta i_r - \lambda_+)
% \end{equation}
% Using this expression for the equation $F_a(i_a,i_r) = 0$, 

% % % %
% We can check the stability of the endemic equilibrium $\bs{i}^* \triangleq (i_a^*,i_r^*) \in (0,x_a]\times(0,1-x_a]$ by calculating the eigenvalues of
% \begin{equation}
%     J(\bs{i}_{\text{end}}) = 
%     \begin{bmatrix}
%         x_a(\mbeta-\abeta)(1-2f)-\mbeta (1-x_a)f - \heal & \mbeta x_a(1-f) \\
%         (1-x_a)(\abetaf + \mbeta(1-f)) & \mbeta(1-x_a-2(1-x_a)f - x_a f)-\abeta x_a(1-f)-\heal
%     \end{bmatrix}
% \end{equation}
% where we denote $f = \frac{\mbeta - \abetax_a - \heal}{\mbeta-\abeta x_a} \in (0,1)$. The eigenvalues of this matrix are computed to be
% \begin{equation}
%     \lambda_+ = (\mbeta-\abeta x_a)(1-2f) - \heal, \quad \lambda_- = -(\mbeta f + \abeta x_a (1-f) + \heal)
% \end{equation}
% It holds that $\lambda_+ < 0$ since $x_a > \frac{\mbeta-\heal}{\abeta}$. Both eigenvalues are negative and therefore $\bs{i}_{\text{end}}$ is locally asymptotically stable. 
% % % %

We will refer to this interior equilibrium as the \emph{endemic} equilibrium, $\bs{i}^* = (i_a^*,i_r^*)$. To establish global asymptotic stability (part 2 of Theorem \ref{thm:IFE_GAS}), we will consider the candidate Lyapunov function $V_R : \Gamma \rightarrow \mbb{R}_{\geq 0}$,
\begin{equation}
    V_R(\bs{i}) \triangleq \max\left\{|i_a - x_a f|,R\cdot|i_r - (1-x_a)f| \right\}
\end{equation}
for some $R > 0$, where $f \triangleq \frac{\lambda_+}{\lambda_+ + \heal} \in (0,1)$. We can write it more explicitly as
\begin{equation}\label{eq:lyapunov_endemic}
    V_R(\bs{i}) =
    \begin{cases}
        x_a f - i_a, &\text{if } \bs{i} \in \Gamma_a^< \\
        i_a - x_a f, &\text{if } \bs{i} \in \Gamma_a^\geq \\
        R((1-x_a)f - i_r), &\text{if } \bs{i} \in \Gamma_r^< \\
        R(i_r - (1-x_a)f), &\text{if } \bs{i} \in \Gamma_r^\geq \\
    \end{cases}
\end{equation}
where the regions are defined as
\begin{equation}
    \begin{aligned}
        \Gamma_a^< &\triangleq \left\{ \bs{i}  : V_R(\bs{i}) = |i_a - x_a f|, \ i_a - x_a f < 0 \right\} \\
        \Gamma_a^\geq &\triangleq \left\{ \bs{i} : V_R(\bs{i}) = |i_a - x_a f|, \ i_a - x_a f \geq 0 \right\} \\
        \Gamma_r^< &\triangleq \left\{ \bs{i} : V_R(\bs{i}) = R|i_r-(1-x_a)f|, i_r - (1-x_a)f < 0 \right\} \\
        \Gamma_r^\geq &\triangleq \left\{ \bs{i} : V_R(\bs{i}) = R|i_r-(1-x_a)f|, i_r - (1-x_a)f \geq 0 \right\} \\
    \end{aligned}
\end{equation}
It holds that $V_R(\bs{i}) \geq 0$ for all $\bs{i} \in \Gamma$, and equality is satisfied if and only if $\bs{i} = \bs{i}^*$. An illustration of the regions is shown in Figure \ref{fig:endemic}.

\begin{figure}[!htbp]
    \centering
    \includegraphics[scale=0.25]{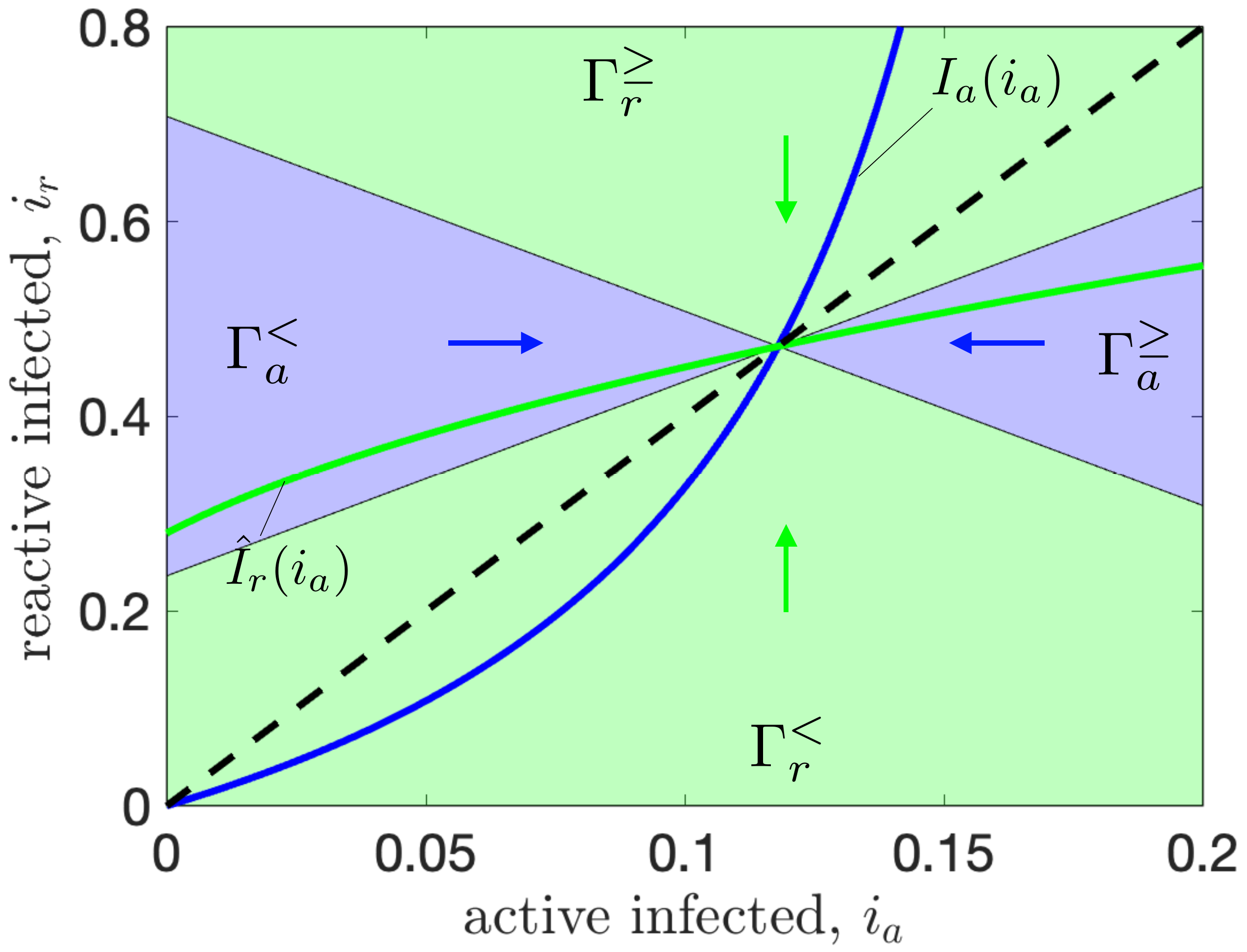}
    \caption{ The regions defined in the Lyapunov function $V_R$ \eqref{eq:lyapunov_endemic}. We set $\mbeta = 0.3$, $\abeta = 0.28$, $\heal = 0.1$, $x_a = 0.2$, which is in the regime $x_a \leq \frac{\mbeta-\heal}{\mbeta + \abeta}$. Here we set $R = 0.5$, which satisfies the condition of \eqref{eq:R_condition}, $0.25 \leq R < \min \{0.8746,0.6143\}$. } 
    \label{fig:endemic}
\end{figure}

In the case $x_a \leq \frac{\mbeta-\heal}{\mbeta+\abeta}$, we will select any $R$ in the range
\begin{equation}\label{eq:R_condition}
    \frac{x_a}{1-x_a} \leq R < \min\left\{\frac{1}{\hat{I}'_r(i_a^*)}, \frac{\mbeta x_a f}{d + \mbeta f (1-x_a)}\right\}.
\end{equation}
The condition $R < 1/\hat{I}'_r(i_a^*)$ ensures that the graph of $\hat{I}_r$, for $i_a \geq x_a f$, is contained in $\Gamma_a^\geq$. The condition $R < \frac{\mbeta x_a f}{d + \mbeta f (1-x_a)}$ ensures that the graph of $\hat{I}_r$, for $i_a < x_a f$, is contained in $\Gamma_a^<$. The condition $\frac{x_a}{1-x_a} \leq R$ ensures that the graph of $I_a$ in $\Gamma$ is contained only in either $\Gamma_r^<$ (for $i_a < x_a f$) or $\Gamma_r^\geq$ (for $i_a \geq x_a f$).
Under \eqref{eq:R_condition}, the following properties hold:
\begin{itemize}
    \item Any $\bs{i} \in \Gamma_a^<$ satisfies $F_a(\bs{i}) > 0$.
    \item Any $\bs{i} \in \Gamma_a^\geq$ satisfies $F_a(\bs{i}) \leq 0$ with equality if and only if $i_a = x_a f$.
    \item Any $\bs{i} \in \Gamma_r^<$ satisfies $F_r(\bs{i}) > 0$.
    \item Any $\bs{i} \in \Gamma_r^\geq$  $F_r(\bs{i}) \leq 0$ with equality if and only if $i_r = (1-x_a)f$.
\end{itemize}
We thus obtain
\begin{equation}
    \frac{dV_R}{dt}(\bs{i}) = 
    \begin{cases}
        -\frac{di_a}{dt}, &\text{if } \bs{i} \in \Gamma_a^< \\
        \frac{di_a}{dt}, &\text{if } \bs{i} \in \Gamma_a^\geq \\
        -R\frac{di_r}{dt}, &\text{if } \bs{i} \in \Gamma_r^< \\
        R\frac{di_r}{dt}, &\text{if } \bs{i} \in \Gamma_r^\geq \\
    \end{cases}
    \ \ \leq 0,
\end{equation}
with equality if and only if $\bs{i} = \bs{i}^*$. Similar arguments can be applied in the case $x_a > \frac{\mbeta-\heal}{\mbeta + \abeta}$. In this case, the arguments hold for the choice $R = 1$. This establishes global asymptotic stability (w.r.t. $\Gamma$) of $\bs{i}^*$ , which concludes the proof of Theorem \ref{thm:IFE_GAS} part 2.

%%%%%%%%%%%%%%%%%%%%%%%%%%%%%%%%%%%%%%%%%%%%%%%%%%%%%%%
\section{Optimal investments in active defense}\label{sec:investment}

We have now fully characterized the global asymptotic behavior of system \ref{eq:SIS_active_dynamics}. It may be summarized by the limiting infected fraction
\begin{equation}
    L(x_a,\abeta) \triangleq \lim_{t \rightarrow \infty} i(t) =
    \begin{cases}
        1 - \frac{\heal}{\mbeta - \abeta x_a}, &\text{if } \abeta x_a < \beta-\heal \\
        0, &\text{else}
    \end{cases}.
\end{equation}
In this section, we consider a system operator that makes  investment decisions to promote the active defense capabilities of the network.  Suppose the operator has a limited monetary budget $M > 0$. It decides to invest an amount of money $a \geq 0$ in increasing the total fraction $x_a$ of active defenders, with a return function  $h : \mbb{R}_{\geq 0} \rightarrow [0,1]$ that is increasing. This reflects a plausible scenario since active defenses are new tools and capabilities which would incur extra costs to install on network nodes. Likewise, it decides to invest an amount $b \geq 0$ in increasing the effectiveness of active defenders, with an expected return function $g : \mbb{R}_{\geq 0} \rightarrow \mbb{R}_{\geq 0}$. This reflects the costs associated with improving active defense tactics, where more advanced tactics incur higher costs (e.g., requiring the active defender to collect and process more data from its neighboring node it is trying to help out). The optimization problem that the operator faces is thus formulated as
\begin{equation}\label{eq:investment_problem}
    \begin{aligned}
        \min_{y = (a,b)} &L(h(a),g(b)) \\
        \text{s.t. } &a,b \geq 0 \\
        & a + b \leq M
    \end{aligned}
\end{equation}
This problem becomes trivial if there is a feasible pair $(a,b)$ that satisfies $g(b)h(a) \geq \beta - \heal$. In this case, the infection can be eradicated. As such, we will consider the cases where no feasible pair achieves this, i.e. $g(b)h(a) < \beta - \heal$ for all $a,b$ with $a + b \leq M$.

We will proceed by defining the following structural properties for the return functions.
\begin{definition}
    We denote the family of functions $\mcal{D}$ as the set of single-variable functions $f: \mbb{R}_{\geq 0} \rightarrow \mbb{R}_{\geq 0}$ that are strictly increasing, twice continuously differentiable, concave, and satisfies $f(0) = 0$.
\end{definition}

For the formulation of optimization problem \eqref{eq:investment_problem}, we place the following assumptions on the return functions $h,g$.
\begin{assumption}\label{assumption:concavity}
The function $g \in \mcal{D}$. The function $h(a)$ is of the form $h(a) = \min\{\hat{h}(a),1 \}$, where $\hat{h} \in \mcal{D}$. 
\end{assumption}
Concavity is a common assumption that describes marginal diminishing returns on investment. For the function $h(a)$, there is a value $t > 0$ for which monetary investments larger than $t$ will convert the entire population to become active defenders. It can be the case that $t = \infty$, in which case $h \in \mcal{D}$. The assumption $h(0) = g(0) = 0$ asserts that zero investment yields zero returns. The following result establishes that there exists a unique optimal investment that solves \eqref{eq:investment_problem}.

\begin{theorem}\label{thm:optimal_investment}
Suppose $g,h$ satisfy the properties of Assumption \ref{assumption:concavity}. Suppose $t \geq M$. Then \eqref{eq:investment_problem} has a unique solution $y^* = (a^*,b^*)$, where $a^* \in (0,M)$ satisfies
    \begin{equation}\label{eq:unique_equation}
        g(M-a^*)h'(a^*) = g'(M-a^*) h(a^*)
    \end{equation}
    and $b^* = M - a^*$. Suppose $t < M$. If $g'(M-t) \geq g(M-t)\hat{h}'(t)$, then the unique solution is given by the same $y^*$.  If $g'(M-t) < g(M-t)\hat{h}'(t)$, then the unique solution is given by $a^* = t$, $b^* = M-t$.
\end{theorem}
\begin{proof}
    First, observe that an investment profile with $a = 0$ and $b > 0$, or $a > 0$ and $b = 0$ cannot be optimal. Thus, an optimal investment must satisfy $a,b > 0$. By complementary slackness, the two multipliers associated with the non-negativity constraints on $a,b$ must be zero.

    Let $\lambda$ be the multiplier associated with the inequality constraint, $a + b \leq M$. Note that if $\lambda = 0$, the first-order conditions become $g(b) = 0$ and $h(a) = 0$, which has the unique solution $a = b = 0$. However, this cannot be an optimal solution. Thus, we suppose $\lambda > 0$. Here, the budget must be met, i.e. $a+b = M$. We obtain the first-order condition
    \begin{equation}\label{eq:FOC}
        H_1(a) = H_2(a)
    \end{equation}
    where $H_1(a) = g(M-a)h'(a)$ and $H_2(a) = g'(M-a) h(a)$. To see this, we must have
    \begin{equation}\label{eq:KKT_ZG}
        \begin{aligned}
            \frac{\heal g(M-a) h'(a)}{(\beta - g(M-a)h(a))^2} = \lambda = \frac{\heal h(a) g'(M-a)}{(\beta - g(M-a)h(a))^2}
        \end{aligned}
    \end{equation}
    Since we are considering the cases where $\beta - g(b)h(a) > \heal > 0$, the denominator may be eliminated.
    
    \noindent\textbf{Case 1: } $M \leq t$. For $a \in (0,M)$, $H_1(a)$ is strictly decreasing because
    \begin{equation}
        H_1'(a) = -g'(M-a)h'(a) + g(M-a) h''(a) < 0,
    \end{equation}
    which follows from the assumptions that $h$ and $g$ are strictly increasing, concave, and twice differentiable (Assumption \ref{assumption:concavity}). For $a \in (0,M)$, $H_2(a)$ is strictly increasing because
    \begin{equation}
        H_2'(a) = -g''(M-a)h(a) + g'(M-a) h'(a) > 0,
    \end{equation}
    which follows similarly. Moreover, $H_1(0) > 0$, $H_1(M) = 0$, and $H_2(0) = 0$, $H_2(M) > 0$. By the intermediate value theorem and by strict monotonicity of $H_1$ and $H_2$, there exists a unique solution $a^*$ to \eqref{eq:FOC}. 
    
    Because Slater's condition holds for \eqref{eq:investment_problem}, the KKT conditions are necessary and sufficient for optimality. Taking  $b^* = M-a^*$ and $\lambda^*$ in \eqref{eq:KKT_ZG}, we have showed that $(a^*,b^*,\lambda^*)$ uniquely satisfy the KKT conditions. 

    \noindent\textbf{Case 2: } $M > t$. In this case, no optimal solution will allocate $a > t$, since any excess allocation here may be invested in $b$. Let us re-formulate \eqref{eq:investment_problem} with the additional constraint $a \leq t$, and denote $\lambda_a \geq 0$ as its associated multiplier. If $g'(M-t) \geq g(M-t)\hat{h}'(t)$, a unique solution $a^* \in (0,t]$ satisfies \eqref{eq:KKT_ZG}, which is the same solution as Case 1. 
    Now suppose $g'(M-t) < g(M-t)\hat{h}'(t)$. Here, note that no solution $a \in (0,t)$ to \eqref{eq:KKT_ZG} exists. Thus, considering $\lambda_a > 0$ ($a = t$), the first-order condition now becomes
    \begin{equation}
        \begin{aligned}
            \heal g(M-t)\partial h(t) &= \heal g'(M-t) \\
            &\quad+ \lambda_a(\beta-g(M-t))^2,
        \end{aligned}
    \end{equation}
    where $\partial h(t) \in [0,\hat{h}'(t)]$ is any sub-derivative of $h$ at $a=t$, and we used the fact that $h(t) = 1$. Taking any $\partial h(t)$ such that $g(M-t)\partial h(t) - g'(M-t) > 0$ (possible by assumption), we obtain
    \begin{equation}
        \lambda_a^* = \frac{\heal(g(M-t)\partial h(t) - g'(M-t))}{(\beta - g(M-t))^2} > 0.
    \end{equation}
    Lastly, we recover $\lambda^* = \frac{\alpha h(t) g'(M-t)}{(\beta - g(M-t))^2}$. Thus, $a^* = t$, $b^* = M-t$, $\lambda^*$, and $\lambda_a^*$ satisfy the KKT conditions, and thus $a^*,b^*$ are solutions of \eqref{eq:investment_problem}.
\end{proof}

\begin{example}
    For linear returns on investment, $h(a) = \min\{c_1 a, 1\}$ and $g(b) = c_2 b$ for constants $c_1,c_2 > 0$. From Theorem \ref{thm:optimal_investment}, if $1/c_1 > M$, we have $a^* = b^* = M/2$. If $M/2 \leq 1/c_1 < M$, we also obtain  $a^* = b^* = M/2$. If $1/c_1 < M/2$, then $a^* = 1/c_1$ and $b^* = M - 1/c_1$.
\end{example}

\begin{example}
    Suppose $h(a) = \frac{a}{a+c_1}$ and $g(b) = \bar{\beta}\frac{b}{b+c_2}$ for constants $c_1,c_2,\bar{\beta} > 0$. From \eqref{eq:unique_equation}, we obtain $a^* = \frac{c_1(M+c_2)}{c_1-c_2}\left[ 1 - \sqrt{1 - \frac{(c_1-c_2)M}{c_1(M+c_2)}}\right] < M$ if $c_1 \neq c_2$, and $a^* = M/2$ if $c_1 = c_2$.
\end{example}

%%%%%%%%%%%%%%%%%%%%%%%%%%%%%%%%%%%%%%%%%%%%%%%%%%%%%%%

\section{A-SIR dynamics of active cyber defense}

In this section, we investigate the impact of active defenders in the SIR (susceptible-infected-recovered) epidemics model. Here, a node that has been cleared of infectious malware obtains permanent protection against any future infection. The protection can be conferred either through reactive defenses (with rate $\heal$), or through active defenses (with rate $\abeta$). As such, we keep track of the five states $s_a(t)$, $s_r(t)$, $i_a(t)$, $i_r(t)$, and $r(t)$. Note that we need not differentiate between active and reactive protected nodes, since they are effectively removed from any further interactions.

\begin{equation}\label{eq:SIP_dynamics}
    \begin{aligned}
        \frac{ds_a}{dt} &= -\mbeta \cdot s_a i \\
        \frac{ds_r}{dt} &= -\mbeta \cdot s_r i \\
        \frac{di_a}{dt} &= \underbrace{\mbeta\cdot s_a i}_{\text{malware infection}}- \underbrace{\abeta\cdot s_a i_a}_{\text{active defense}} - \underbrace{\heal i_a}_{\text{reactive defense}} \\
        \frac{di_r}{dt} &= \mbeta\cdot s_r i - \abeta\cdot s_a i_r - \heal i_r \\
        \frac{dr}{dt} &= \abeta\cdot s_a i + \heal i \\
    \end{aligned}\tag{A-SIR}
\end{equation}
with initial conditions $(s_{a0},s_{r0},i_{a0},i_{r0}) \in [0,1]^4$ that satisfy $s_{a0}+s_{r0}+i_{a0}+i_{r0}=1$. Here, we are assuming that no node is initially fully protected from the malware, meaning $r(0) = 0$. We consider the class of initial value problems with the above dynamics \eqref{eq:SIP_dynamics} parameterized by a fixed fraction of initially infected nodes $i_0$ and the fraction of active defenders in the population, namely $s_{a0}$. Note that since no node is initially fully protected, we have $s_0 = s_{a0} + s_{r0} = 1 - i_0$. Also, observe that the roles of infected active or infected reactive nodes are indistinguishable (both types infect susceptibles at the same rate, and both attain protection at the same rates). Indeed, the total infected fraction $i$ depends only on $i$, but not on $(i_a,i_r)$:
\begin{equation}
    \frac{di}{dt} = \mbeta\cdot si - \abeta \cdot s_a i - \heal i.
\end{equation}
Therefore, we only consider the initial fraction of active susceptible nodes, namely $s_{a0}$.

The following Theorem characterizes the peak infection level for any initial value problem of the A-SIP dynamics.

\begin{theorem}
    Consider any initial value problem of the A-SIP dynamics, and denote by $i(t)$, $t \geq 0$ the resulting state trajectory for the total infected fraction of nodes. Then $\ipeak \triangleq \max_{t\geq 0} i(t)$ is characterized by
    \begin{equation}\label{eq:peak}
        \ipeak = 
        \begin{cases}
            1 - \frac{\heal}{\mbeta} - \frac{\abeta}{\mbeta}s_{a0} + \frac{\heal}{\mbeta}\log\frac{\heal}{\mbeta s_0 - \abeta s_{a0}}, &\text{if } s_{a0} < \frac{\mbeta s_0 - \heal}{\abeta} \\
            i_0, &\text{if } s_{a0} \geq \frac{\mbeta s_0 - \heal}{\abeta}
        \end{cases}.
    \end{equation}
\end{theorem}
\begin{proof}
    We can characterize the trajectory of $i$ as a function of $s_a$ as follows. Observe that
    \begin{equation}
        \frac{ds_a}{ds_r} = \frac{s_a}{s_r}.
    \end{equation}
    Integrating this equation, we obtain the relation
    \begin{equation}
        s_a(s_r) = \frac{s_a(s_{r0})}{s_{r0}}\cdot s_r = \frac{s_{a0}}{s_0-s_{a0}}\cdot s_r,
    \end{equation}
    where we denote $s_0 = s_{a0} + s_{r0}$. Using this relation, we may write
    \begin{equation}
        \frac{di}{dt} = \left(\mbeta\frac{s_0}{s_{a0}} - \abeta \right) s_a i - \heal i,
    \end{equation}
    and subsequently,
    \begin{equation}\label{eq:di_dsa}
        \frac{di}{ds_a} = -\frac{A}{\mbeta} + \frac{\heal}{\mbeta s_a},
    \end{equation}
    where $A := \mbeta\frac{s_0}{s_{a0}} - \abeta$. Observe that function $i(s_a)$ is concave in $s_a \in (0,\infty)$. The peak of the trajectory $i(s_a)$ is obtained by setting the above equation to $0$, and solving for $s_a$. 
    
    In the case $A \leq 0$, or equivalently $s_{a0} \geq \frac{\mbeta s_0}{\abeta}$, function $i(s_a)$ does not have a peak for any $s_a \geq 0$. Therefore, the infected fraction $i(t)$ is monotonically decreasing in time until eradication, and hence $\max_{t\geq 0} i(t) = i_0$.

    In the case $A > 0$, function $i(s_a)$ peaks  precisely at $s_a^* = \frac{\heal}{A}$. Since it must hold that $s_a \in [0,s_{a0})$ throughout the entire trajectory, the infected fraction $i(t)$ peaks after the initial time if and only if $\frac{\heal}{A} < s_{a0}$. This is equivalent to the condition $s_{a0} < \frac{\mbeta s_0 - \heal}{\abeta}$. To determine the peak infection level, we integrate \eqref{eq:di_dsa} to obtain
    \begin{equation}
        i(s_a) = i_0 - \frac{A}{\mbeta}(s_a-s_{a0}) + \frac{\heal}{\mbeta}\log\frac{s_a}{s_{a0}}.
    \end{equation}
    Evaluating at $s_a^* = \frac{\heal}{A}$, the peak value is 
    \begin{equation}
        i^* = 1 - \frac{\heal}{\mbeta} - \frac{\abeta}{\mbeta}s_{a0} + \frac{\heal}{\mbeta}\log\frac{\heal}{\mbeta s_0 - \abeta s_{a0}}
    \end{equation}
    where we have used the fact that $s_0 + i_0 = 1$. Lastly, if $\frac{\heal}{A} \geq s_{a0}$, or equivalently $s_{a0} \geq \frac{\mbeta s_0 - \heal}{\abeta}$, then the infected fraction $i(t)$ is monotonically decreasing, and hence $i^* = i_0$.
\end{proof}
The above result indicates that the infection level will not exceed $i_\text{pk}$ for any time $t \geq 0$. Thus, if the objective of a system operator is to ensure the network never reaches an infectivity level above a certain desired threshold $\tau$, then equation \eqref{eq:peak} can help specify design parameters (e.g. $\beta_a$, $s_{a0}$, $\alpha$) on defense that meet this requirement. A precise characterization will be left for future work. 

%%%%%%%%%%%%%%%%%%%%%%%%%%%%%%%%%%%%%%%%%%%%%%%%%%%%%%%
\section{Conclusion}
Active cyber defense is an emerging technology. We have proposed a novel active cyber defense model based on the epidemiological SIS population dynamics. We fully characterized the behavior of the dynamics, establishing global asymptotic stability of infection-free and endemic fixed points. We show that deploying active cyber defenses has an impact on the epidemic threshold, unlike other mitigation approaches studied in epidemic models such as reactive social distancing. We further leverage the characterization to determine optimal investments in active cyber defense. As a side-product, we also characterize the effect of deploying active cyber defense in an SIR model. We hope this study will inspire more investigations on the effectiveness of active cyber defense, which is a new paradigm in cyber defense that could be a game-changer in cybersecurity.

\bibliographystyle{IEEEtran}
\bibliography{sources}

%%%%%%%%%%%%%%%%%%%%%%%%%%%%%%%%%%%%%%%%%%%%%%%%%
%=== ================APPENDIX ============================================
%%%%%%%%%%%%%%%%%%%%%%%%%%%%%%%%%%%%%%%%%%%%%%%%%
% \appendix
% \input{appendix}

%----- Bibliography ----------------

\end{document}